\documentclass[11pt]{article}
\usepackage{times}
\usepackage{amsmath,amsfonts}
\usepackage{epsfig,amssymb,amstext,xspace,theorem} 
\usepackage{graphicx,color}
\usepackage{verbatim}
\usepackage{multirow}
\usepackage[shortlabels]{enumitem}
\usepackage{comment}
\usepackage{calc}

\newtheorem{theorem}{Theorem}[section] 
\newtheorem{lemma}[theorem]{Lemma}

\newtheorem{corollary}[theorem]{Corollary}

{\theorembodyfont{\rmfamily} 
\newtheorem{remark}[theorem]{Remark}

\newtheorem{defn}[theorem]{Definition} 
}

\def\blksquare{\rule{2mm}{2mm}} 
\def\qedsymbol{\blksquare}
\newcommand{\qedhere}{\tag*{\qedsymbol}}
\newcommand{\bg}[1]{\medskip\noindent{\bf #1}}
\newcommand{\ed}{{\hfill\qedsymbol}\medskip} 
\newenvironment{proofbox}[1][Proof]{\bg{{#1}\ :\ }}{\ed}
\newenvironment{proof}{\begin{proofbox}}{\end{proofbox}}
\newenvironment{proofnobox}{\bg{Proof\ :\ }}{\medskip}

\newcommand{\R}{\ensuremath{\mathbb R}}

\newcommand{\F}{\ensuremath{\mathcal F}}
\newcommand{\D}{\ensuremath{\mathcal D}} 
\newcommand{\T}{\ensuremath{\mathcal T}}

\newcommand{\Pc}{\ensuremath{\mathcal P}}
\newcommand{\Qc}{\ensuremath{\mathcal Q}} 
\newcommand{\OPT}{\ensuremath{\mathit{OPT}}}

\newcommand{\sm}{\ensuremath{\setminus}} 
\newcommand{\es}{\ensuremath{\emptyset}}

\newcommand{\ceil}[1]{\ensuremath{\left\lceil#1\right\rceil}}

\newcommand{\poly}{\operatorname{\mathsf{poly}}}

\newcommand{\e}{\ensuremath{\epsilon}} 

\newcommand{\gm}{\ensuremath{\gamma}}

\newcommand{\sse}{\subseteq}

\newcommand{\vrp}{\ensuremath{\mathsf{VRP}}\xspace}
\newcommand{\mlp}{\ensuremath{\mathsf{MLP}}\xspace}

\newcommand{\rvrp}{\ensuremath{\mathsf{RVRP}}\xspace}

\newcommand{\mst}{\ensuremath{\mathsf{MST}}\xspace}
\newcommand{\tsp}{\ensuremath{\mathsf{TSP}}\xspace}
\newcommand{\atsp}{\ensuremath{\mathsf{ATSP}}\xspace}

\newcommand{\iopt}{\ensuremath{O^*}} 
 
\newcommand{\into}{\ensuremath{\mathrm{in}}}
\newcommand{\out}{\ensuremath{\mathrm{out}}}

\newcommand{\red}{\ensuremath{\mathsf{red}}}
\newcommand{\dist}{\ensuremath{D}}
\newcommand{\reg}{\ensuremath{\mathsf{reg}}}

\newcommand{\ld}{\ensuremath{\lambda}} 
\newcommand{\kp}{\ensuremath{\kappa}}
\newcommand{\al}{\ensuremath{\alpha}} 

\newcommand{\tht}{\ensuremath{\theta}} 
\newcommand{\dt}{\ensuremath{\delta}}
 
\newcommand{\sg}{\ensuremath{\sigma}}

\newcommand{\tP}{\ensuremath{\tilde P}} 
\newcommand{\hP}{\ensuremath{\hat P}}

\newcommand{\ox}{{\overline x}\xspace}

\newcommand{\rewd}{\ensuremath{\mathsf{rewd}}\xspace}
\newcommand{\ptp}{\ensuremath{\mathsf{P2P}}\xspace}

\title{Compact, Provably-Good LPs for Orienteering and Regret-Bounded Vehicle Routing
\footnote{A preliminary version~\cite{FriggstadS17} appeared in the Proceedings of the
  19th Conference on Integer Programming and Combinatorial Optimization, 2017.}} 
\author{ 
    Zachary Friggstad\thanks{{\tt zacharyf@ualberta.ca}.  
    Dept. of Computer Science, Univ. Alberta, Edmonton, AB T6G 2E8.
    Supported by the Canada Research Chairs program and an NSERC Discovery grant.}  
\and
    Chaitanya Swamy\thanks{{\tt cswamy@uwaterloo.ca}.  
    Dept. of Combinatorics and Optimization, Univ. Waterloo, Waterloo, ON N2L 3G1. 
    Supported in part by NSERC grant 327620-09 and an NSERC Discovery Accelerator
    Supplement Award.}
}
\date{}

\begin{document}

\maketitle

\vspace*{-3ex}

\begin{abstract}
We develop polynomial-size LP-relaxations for {\em orienteering} and the 
{\em regret-bounded vehicle routing problem} (\rvrp) and devise suitable LP-rounding 
algorithms that lead to various new insights and approximation results for these problems.  
In orienteering, the goal is to find a maximum-reward $r$-rooted path, possibly ending at
a specified node, of length at most some given budget $B$. 
In \rvrp, the goal is to find the minimum number of $r$-rooted paths of {\em regret} at
most a given bound $R$ that cover all nodes, where the regret of an $r$-$v$ path is its
length $-$ $c_{rv}$. 

For {\em rooted orienteering}, we introduce a natural bidirected LP-relaxation and obtain
a simple $3$-approximation algorithm via LP-rounding. This is the {\em first LP-based}
guarantee for this problem. We also show that {\em point-to-point} (\ptp) 
{\em orienteering} can be reduced to a regret-version of rooted orienteering at the
expense of a factor-2 loss in approximation. 
For \rvrp, we propose two compact LPs that lead to significant 
improvements, in both approximation ratio and running time, over the approach
in~\cite{FriggstadS14}. 
One of these is a natural modification of the LP for rooted orienteering;  
the other is an unconventional formulation that is motivated by certain structural
properties of an \rvrp-solution, which leads to a $15$-approximation algorithm for \rvrp.  
\end{abstract}

\section{Introduction}
Vehicle-routing problems (\vrp{}s) constitute a broad class of optimization
problems that find a wide range of applications and have been widely studied in the
Operations Research and Computer Science literature (see, 
e.g.,~\cite{HaimovichK85,TothV02,CharikarR98,BlumCKLMM07,BansalBCM04,ChekuriKP12}). 
Despite this extensive study, we have rather limited understanding of LP-relaxations for
\vrp{}s (with \tsp and the minimum-latency problem, to a lesser extent, being exceptions),
and this has been an impediment in the design of approximation algorithms for these
problems. 

Motivated by this gap in our understanding, we investigate whether one can develop
polynomial-size (i.e., compact) LP-relaxations with good integrality gaps for 
\vrp{}s, focusing on the fundamental {\em orienteering}
problem~\cite{GoldenLV87,BlumCKLMM07,ChekuriKP12} and the related  
{\em regret-bounded vehicle routing problem} (\rvrp)~\cite{BockGKS11,FriggstadS14}. 
In {\em orienteering}, we are given rewards associated with clients located in a metric
space, a length bound $B$, a start, and possibly end, location for the vehicle, and we
seek a route of length at most $B$ that gathers maximum reward.
This problem frequently arises as a subroutine when solving \vrp{}s, both in approximation 
algorithms---e.g., for minimum-latency problems
(\mlp{}s)~\cite{BC+94,FakcharoenpholHR07,ChakrabartyS16,PostS15},   
TSP with time windows~\cite{BansalBCM04}, 
\rvrp~\cite{BockGKS11,FriggstadS14}---%
as well as in computational methods where orienteering 
corresponds to the ``pricing'' problem encountered in solving set covering/partitioning
LPs (a.k.a configuration LPs) for \vrp{}s via a column-generation or branch-cut-and-price 
method.
In \rvrp, we have a metric space $\{c_{uv}\}$ on client locations, a start
location $r$, and a {\em regret} bound $R$. The regret of a path $P$ starting at $r$ and  
ending at location $v$ is $c(P)-c_{rv}$. The goal in \rvrp is to find a minimum
number of $r$-rooted paths of regret at most $R$ that visit all clients.

\paragraph{Our contributions.}
We develop polynomial-size LP-relaxations for orienteering and \rvrp and devise
suitable rounding algorithms for these LPs, which lead to various new insights and
approximation results for these problems.  

In Section~\ref{sec:rt_orient}, we introduce a natural, compact LP-relaxation for 
{\em rooted orienteering}, wherein only the vehicle start node is specified, and
design a simple rounding algorithm to convert an LP-solution to an integer solution losing 
a factor of at most 3 in the objective value. This is the {\em first LP-based}
approximation guarantee for orienteering. In contrast, all other approaches for
orienteering 
utilize dynamic programming (DP) to stitch together suitable subpaths. 

In Section~\ref{sec:ptp_orient}, we consider the more-general 
{\em point-to-point} (\ptp) {\em orienteering} problem, where both the start and end nodes
of the vehicle are specified. We present a novel reduction showing that 
\ptp-orienteering can be {\em reduced} to a {\em regret-version} of rooted orienteering,  
wherein the length bound is replaced by a 
{\em regret bound}, incurring a factor-2 loss (Theorem~\ref{p2predn}).  
No such reduction to a rooted problem was known previously, and all known algorithms for
\ptp-orienteering rely on approximations to suitable \ptp-path problems.  
Typically, 
constraining a \vrp by requiring that routes include a fixed node $t$ 
causes an increase in the route lengths of the unconstrained problem (as we need to 
attach $t$ to the routes); this would violate the length bound in orienteering, 
but, notably, we devise a way to avoid this in our reduction.    
We believe that the insights gained from our reduction may find further application.
Our results for rooted orienteering translate to the regret-version of orienteering,
and combined with the above reduction, give a compact LP for \ptp-orienteering having
{integrality gap at most $6$.}

Although 
we do not improve the current-best approximation factor of $(2+\e)$ for
orienteering~\cite{ChekuriKP12}, we believe that our LP-based approach is
nevertheless appealing for various reasons.
First, our LP-rounding algorithms 
are quite simple, and arguably, simpler than the DP-based approaches
in~\cite{BlumCKLMM07,ChekuriKP12}. 
Second, our LP-based approach offers the promising possibility 
that, by leveraging the key underlying ideas, 
one can obtain strong, compact LP-relaxations for other problems that utilize
orienteering. Indeed, we already present evidence of such benefits by 
showing 
in Section~\ref{orientreglp} that our LP-insights for rooted orienteering yield a
compact, provably-good LP for \rvrp. 
(We remark that various configuration LPs considered for \vrp{}s  
give rise to \ptp-orienteering as the dual-separation problem, and utilizing our compact    
orienteering-LP in the dual could yield another way of obtaining a compact LP.)
Finally, LP-based insights often tend to be powerful 
and have the potential to result in both improved guarantees, and algorithms for variants
of the problem.   
In fact, we suspect that our orienteering LPs, \eqref{lp:root_or}, \eqref{lp:p2p_or},
are better than what we have accounted for, 
and believe that they are a promising means of 
{improving the state-of-the-art for orienteering.} 

Section~\ref{sec:rvrp} considers \rvrp, and proposes two compact LP-relaxations for 
\rvrp and corresponding rounding algorithms. Our LP-based algorithms not only yield
improvements over the current-best $28.86$-approximation for \rvrp~\cite{FriggstadS14},
but also result in substantial savings in running time compared  
to the algorithm in~\cite{FriggstadS14}, which involves solving a configuration LP 
(with an exponential number of path variables) using the $\Omega(n^{1/\e})$-time
$(2+\e)$-approximation algorithm for orienteering in~\cite{ChekuriKP12} as a subroutine.   
The first LP for \rvrp is a natural modification of our LP for rooted orienteering, which
we show has integrality gap at most $27$ (Theorem~\ref{rlp1thm}). 
In Section~\ref{newreglp}, we formulate a rather atypical LP-relaxation
\eqref{rlp2} for \rvrp by exploiting certain key structural insights for \rvrp. 
We observe that an \rvrp-solution can be regarded as a collection of distance-increasing 
rooted paths covering some {\em sentinel} nodes $S$ and a low-cost way of connecting the
remaining nodes to $S$, and our LP aims to find the best such solution. 
We design a rounding algorithm for this LP that leads to a 15-approximation algorithm for
\rvrp, which is a significant improvement over the guarantee obtained
in~\cite{FriggstadS14}.

Finally, in Section~\ref{minregtsp}, we observe that our techniques imply that the
integrality gap of a Held-Karp style LP for the {\em asymmetric-TSP} (\atsp) {\em path} 
problem is 2 for the class of asymmetric metrics induced by the regret objective.   

\smallskip
To give an overview of our techniques, a key tool that we use in our rounding algorithms,  
which also motivates our LP-relaxations, is an arborescence-packing result
of~\cite{BangjensenFJ95} showing that    
an $r$-preflow $x\in\R_+^A$ in a digraph $D=(N,A)$ (i.e., 
$x\bigl(\dt^\into(v)\bigr)\geq x\bigl(\dt^\out(v)\bigr)\ \forall v\neq r$)
dominates a weighted collection of $r$-rooted (non-spanning) out-arborescences
(Theorem~\ref{arbpack}). 
An $r$-preflow $x$ in the bidirected version of our metric, $D$, is a natural
relaxation of an $r$-rooted path, and the $r\leadsto u$ connectivity under $x$
abstracts whether $u$ lies on this path. This leads to our LP \eqref{lp:root_or}
for (rooted) orienteering. 
The idea behind the rounding is that 
if we know the node $v$ on the optimum path with maximum $c_{rv}$ value, then we can 
enforce that our the LP-preflow $x$ is consistent with $v$. Hence, we can decompose $x$ 
into arborescences containing $v$ of average length at most $B$,
which yield $r$-$v$ paths of average 
{\em regret} at most $2(B-c_{rv})$. These in turn can be converted 
(see Lemma~\ref{avg2max}) into a weighted collection of paths of total weight at most $3$,
where each path has regret at most $B-c_{rv}$ and ends at some node $u$ with 
$c_{ru}\leq c_{rv}$; returning the maximum-reward path in this collection yields a
$3$-approximation.

\paragraph{Related work.}
The orienteering problem seems to have been first defined in~\cite{GoldenLV87}. Blum et
al.~\cite{BlumCKLMM07} gave the first $O(1)$-factor approximation for rooted
orienteering. They obtained an approximation ratio of $4$, which was generalized to
\ptp-orienteering, and improved to $3$~\cite{BansalBCM04} and then to
$2+\e$~\cite{ChekuriKP12}.   

Orienteering is closely related to the {\em $k$-\{stroll, \mst, \tsp{}\}} problems, which
seek a minimum-cost rooted \{path,tree,tour\} respectively spanning at least $k$ nodes
(so the roles of objective and constraint are interchanged). $k$-\mst has a rich
history of study that culminated in a factor-$2$ approximation for both $k$-\mst and
$k$-\tsp~\cite{Garg05}. Chaudhuri et al.~\cite{ChaudhuriGRT03} obtained a
$(2+\e)$-approximation algorithm for $k$-stroll. They also showed that for certain  
values of $k$, one can obtain a tree spanning $k$ nodes and containing two specified nodes
$r,t$, of cost at most the cheapest $r$-$t$ path spanning $k$ nodes. In particular, this
holds for $k=n$, and yields an alternative way of obtaining a $2$-approximation algorithm
for the {\em minimum-regret TSP-path} problem considered in Section~\ref{minregtsp}.
The orienteering algorithms in~\cite{BlumCKLMM07,BansalBCM04,ChekuriKP12} are all 
based on 
first obtaining suitable subpaths
by approximating the {\em min-excess path} problem using a $k$-stroll algorithm as a
subroutine, and then stitching together these subpaths via a DP. (For a rooted 
path, the notions of excess and regret coincide; 
we use the term regret as it is more in line with the terminology used in the
vehicle-routing literature~\cite{SpadaBL05,ParkK10}.) 

The use of regret as a vehicle-routing objective seems to have been first considered
in~\cite{SpadaBL05}, who present various heuristics, 
and \rvrp is sometimes referred to as the {\em schoolbus problem} in the
literature~\cite{SpadaBL05,ParkK10,BockGKS11}. Bock et al.~\cite{BockGKS11} were the first
to consider \rvrp from an approximation-algorithms perspective. They obtain approximation
factors of $O(\log n)$ for general metrics and $3$ for tree metrics. Subsequently,
Friggstad and Swamy~\cite{FriggstadS14} gave the first constant-factor approximation
algorithm for \rvrp, obtaining a $28.86$-approximation via an LP-rounding procedure for a 
configuration LP.

\section{Preliminaries and notation}
Both orienteering and \rvrp involve 
a complete undirected graph  
$G=(\{r\}\cup V,E)$, where $r$ is a distinguished root (or depot) node, and metric edge
costs $\{c_{uv}\}$. 
Let $n=|V|+1$.
We call a path $P$ in $G$ rooted if it begins at $r$. 
We always think of the nodes on $P$ as being ordered in
increasing order of their distance along $P$ from $r$, 
and directing $P$ away from $r$ means that we direct each edge $uv\in P$ from $u$ to
$v$ if $u$ precedes $v$ (under this ordering). 
We use $\dist_v$ to denote $c_{rv}$ for all $v\in V\cup\{r\}$. 
Let $\T$ denote the collection of all $r$-rooted trees in $G$.
For a vector $d\in\R^E$, and a subset $F\sse E$, we use $d(F)$ to denote 
$\sum_{e\in F}d_e$. Similarly, for a vector $d\in\R^V$ and $S\sse V$, we use $d(S)$ to
denote $\sum_{v\in S}d_v$.

\paragraph{Regret metric and \rvrp.}
For every ordered pair $u,v \in V \cup \{r\}$, define the {\em regret distance}
(with respect to $r$) to be $c^{\reg}_{uv}:=\dist_u+c_{uv}-\dist_v$. 
The regret distances $\{c^{\reg}_{uv}\}$ form an asymmetric metric that we call the 
{\em regret metric}. 
The regret of a node $v$ lying on a rooted path $P$ is given by
$c^{\reg}_P(v):=c_P(v)-\dist_v=(c^{\reg}$-length of the $r$-$v$ portion of $P$), where
$c_P(v)$ is the length of the $r$-$v$ subpath of $P$. Define the regret of $P$ to be
$c^\reg(P)$, which is also the regret of the end-node of $P$.
Observe that $c^\reg(Z)=c(Z)$ for any cycle $Z$. 
{We utilize the following results from~\cite{FriggstadS14}.}

\begin{lemma}[\cite{FriggstadS14}] \label{cor:average} \label{avg2max}
Let $R\geq 0$. Given rooted paths $P_1, \ldots, P_k$ with total regret $\alpha R$, 
we can efficiently find at most $k+\al$ rooted paths, each having regret at most $R$, that
cover $\bigcup_{i=1}^k P_i$. 
\end{lemma}

\begin{theorem}[\cite{FriggstadS14}] \label{fs-rvrpthm}
Let $x=(x_P)_{P\in\Pc}$ be a weighted collection of rooted paths such that 
$\sum_{P\in\Pc: v\in P}x_P\geq 1$ for all $v\in V$. Let $R\geq 0$ be some given parameter. 
Let $k=\sum_{P\in\Pc}x_P$ and $\sum_{P\in\Pc}c^\reg(P)x_P=\al R$. Then, for
any $\tht\in(0,1)$, we can round $x$ to obtain a collection of at most 
$\bigl(\frac{6}{1-\tht}+\frac{1}{\tht}\bigr)\al+\ceil{\frac{k}{\tht}}$
rooted paths each of regret at most $R$ that cover all nodes in $V$.
\end{theorem}

\paragraph{Preflows and arborescence packing.}
Let $D=(\{r\}\cup V,A)$ be a digraph. We say that a vector $x\in\R_+^A$ is an 
{\em $r$-preflow} if $x\bigl(\dt^\into(v)\bigr)\geq x\bigl(\dt^\out(v)\bigr)$ for all
$v\in V$. When $r$ is clear from the context, we simply say preflow. 
A key tool that we exploit is an arborescence-packing result of Bang-Jensen et
al.~\cite{BangjensenFJ95} showing that we can decompose a preflow into
out-arborescences rooted at $r$, 
and this can be done in polytime~\cite{PostS15}. By an 
{\em out-arborescence rooted at $r$}, we mean a subgraph  
$B$ whose undirected version is a tree containing $r$, and where every node spanned by $B$  
except $r$ has exactly one incoming arc in $B$.  

\begin{theorem}[\cite{BangjensenFJ95,PostS15}] \label{arbpack}
Let $D=(\{r\}\cup V,A)$ be a digraph and $x\in\R_+^A$ be a preflow. Let
$\ld_v:=\min_{\{v\}\sse S\sse V}x\bigl(\dt^\into(S)\bigr)$ be the $r\leadsto v$
``connectivity'' in $D$ under capacities $\{x_a\}_{a\in A}$. Let $K>0$ be
rational. We can obtain out-arborescences $B_1,\ldots,B_q$ rooted at $r$, and rational
weights $\gm_1,\ldots,\gm_q\geq 0$ such that $\sum_{i=1}^q\gm_i=K$, 
$\sum_{i:a\in B_i}\gm_i\leq x_a$ for all $a\in A$, and 
$\sum_{i:v\in B_i}\gm_i=\min\bigl\{K,\ld_v\}$ for all $v\in V$. Moreover, such a
\mbox{decomposition can be computed in time $\poly(|V|,\text{size of $K$})$.}
\end{theorem}

\section{Rooted orienteering} \label{sec:rt_orient}
In the {\em rooted orienteering} problem, we have a complete undirected graph  
$G=(\{r\}\cup V,E)$, metric edge costs $\{c_{uv}\}$, a distance bound $B \geq 0$, and
nonnegative node rewards $\{\rho(v)\}_{v \in V}$. 
The goal is to find a 
rooted path with cost at most $B$ that collects the maximum reward. 
Whereas all current approaches for orienteering rely on a dynamic program to
stitch together suitable subpaths, we present a simple LP-rounding-based $3$-approximation
{algorithm for rooted orienteering.} 

Let $D=(\{r\}\cup V,A)$ denote the bidirected version of $G$, 
where both $(u,v)$ and $(v,u)$ get cost $c_{uv}$.
To introduce our LP and our rounding algorithm, first suppose that we know a node $v$ on
the optimum path that has maximum distance $\dist_v$ among all nodes on the optimum path.  
In our relaxation, we model the path as one unit of flow $x\in\R_+^A$ that exits $r$,
visits only nodes $u$ with $\dist_u\leq\dist_v$ and $v$ to an extent of 1,
and has cost at most $B$. Since we do not know the endpoint of our path, we relax $x$ to
be a preflow. Letting $z^v_u$ denote the $r\leadsto u$ connectivity (under capacities
$\{x_a\}$), the reward earned by $x$ is $\rewd(x):=\sum_{u\in V}\rho(u)z^v_u$.

Our rounding procedure is based on the insight that Theorem~\ref{arbpack} 
allows us to view $x$ as a convex combination of  
arborescences, which we regard as $r$-rooted trees in $G$. Converting each tree into an
$r$-$v$ path (by standard doubling and shortcutting), we get a convex combination of
rooted paths of average reward $\rewd(x)$, and average cost at most $2B-\dist_v$,
and hence average $c^{\reg}$-cost at most $2(B-\dist_v)$. Applying Lemma~\ref{avg2max} to
this collection, we then obtain a weighted collection of rooted paths of total weight at
most 3 earning the same total reward, where each path has regret at most $B-\dist_v$, and
hence, cost at most $B$ (since it ends at some node $u$ with $\dist_u\leq\dist_v$). Thus,
the maximum-reward path in this collection yields a feasible solution with reward at least
$\rewd(x)/3$. 

Finally, we circumvent the need for ``guessing'' $v$ by using variables
$z^v_v$ to indicate if $v$ is the maximum-distance node on the optimum path. We impose
that we have a preflow $x^v$ of value $z^v_v$ that visits $v$ to an extent of $z^v_v$,
and only visits nodes $u$ with $\dist_u\leq\dist_v$, and $z^v_u$ is now the $r\leadsto u$
{connectivity under capacities $x^v$.} (Note that $r\notin V$.)
\begin{alignat}{3}  
\max & \quad & \sum_{u,v\in V}\rho(u)z^v_u & \tag{{R-O}} \label{lp:root_or} \\
\text{s.t.} && 
x^v(\dt^\into(u)\bigr) & \geq x^v\bigl(\dt^\out(u)\bigr) \qquad && 
\forall u,v\in V \label{pref1} \\
&& x^v\bigl(\dt^\into(u)\bigr) & = 0 \qquad && 
\forall u,v\in V: \dist_u>\dist_v \label{pref2} \\
&& x^v\bigl(\delta^{\into}(S)\bigr) & \geq z^v_u \qquad 
&& \forall v\in V, S\subseteq V, u\in S \label{con} \\
&& \sum_{a\in A}c_ax^v_a & \leq B z^v_v  && \forall v \in V \label{cbound} \\
&& x^v\bigl(\dt^\out(r)\bigr) & =z^v_v \quad \forall v\in V,
&& \sum_{v} z^v_v = 1, \quad x,z \geq 0. \notag
\end{alignat}
This formulation can be converted to a compact LP by 
introducing flow variables $f^{u,v}=\{f^{u,v}_a\}_{a\in A}$, and encoding the cut
constraints \eqref{con} by imposing that $f^{u,v}\leq x^v$, and that
$f^{u,v}$ sends $z^v_u$ units of flow from $r$ to $u$. 
Observe that: (a) if $\dist_u>\dist_v$ then 
$z^v_u\leq x^v\bigl(\dt^\into(u)\bigr)=0$; 
(b) we have $z^{v}_u\leq x^v\bigl(\dt^\into(V)\bigr)=x^v\bigl(\dt^\out(r)\bigr)=z^{v}_v$ 
for all $u,v$. 
Let $(x^*,z^*)$ be an optimal solution to \eqref{lp:root_or}, of value $\OPT$. 

\begin{theorem}\label{thm:rt_gap}
We can round $(x^*,z^*)$ to a rooted-orienteering solution of value at least $\OPT/3$. 
\end{theorem}

\begin{proofnobox}
For each $v$ with $z^{*v}_v > 0$ we apply Theorem \ref{arbpack} with $K = z^{*v}_v$ to obtain
$r$-rooted out-arborescences, which we view as rooted trees in $G$, 
and associated nonnegative weights 
$\{\gm^v_T\}_{T\in\T}$; recall that $\T$ is the collection of all 
$r$-rooted trees. So we have
$\sum_T\gm^v_T=z^{*v}_v$, $\sum_T\gm^v_Tc(T)\leq\sum_a c_ax^{*v}_a\leq Bz^{*v}_v$, 
and $\sum_{T:u\in T}\gm^v_T\geq z^{*v}_u$ for all $u\in V$.
Note that for every $T$ with $\gm_v^T>0$, we have $v\in T$, and $\dist_u\leq\dist_v$ for
all $u\in T$ (as otherwise, we have $x^{*_v}\bigl(\dt^\into(u)\bigr)=0$).
For every $v$ and every tree $T$ with $\gm^v_T>0$, we do the following. 
First, we double the edges not lying on the $r$-$v$ path of $T$
and shortcut to obtain a simple $r$-$v$ path $P^v_T$. 
So 
\begin{equation} 
\sum_T \gm^v_Tc^{\reg}(P^v_T)\leq 2\sum_T\gm^v_T\bigl(c(T)-\dist_v\bigr)
= 2z^{*v}_v(B-\dist_v). \label{orineq1}
\end{equation}
Next, 
we use Lemma~\ref{avg2max} with regret-bound $B-\dist_v$ to break $P^v_T$ into
a collection $\Pc^v_T$ of at most $1 + \frac{c^{\reg}(P^v_T)}{B-\dist_v}$
rooted paths, each having $c^{\reg}$-cost at most $B-\dist_v$. 
Note that if $B=\dist_v$, then $c^{\reg}(P^v_T) = 0$, and we use the convention that
$0/0=0$, so $|\Pc^v_T|=1$ in this case.  
Each path in $\Pc^v_T$ ends at a vertex $u$ with $\dist_u\leq\dist_v$, 
so its $c$-cost is at most $B$. 
Now, for all $v\in V$, we have
\begin{align} 
\sum_T\gm^v_T\sum_{P\in\Pc^v_T}\rho(P) & =\sum_T\gm^v_T\rho(P^v_T)\geq\sum_u\rho(u)z^{*v}_u 
\label{orineq2} \\
\sum_T\gm^v_T|\Pc^v_T| & \leq\sum_T\gm^v_T\Bigl(1+\tfrac{c^{\reg}(P^v_T)}{B-\dist_v}\Bigr)
\leq z^{*v}_v+2z^{*v}_v=3z^{*v}_v \label{orineq3}
\end{align}
where the last inequality in \eqref{orineq3} follows from \eqref{orineq1}.
Therefore, the maximum-reward path in $\bigcup_{v,T:\gm^v_T>0}\Pc^v_T$ earns reward at least
\vspace{-2ex}
\begin{equation}
\Bigl({\sum_{v,T}\gm^v_T\sum_{P\in\Pc^v_T}\rho(P)}\Bigr)\Bigl/
\Bigl({\sum_{v,T}\gm^v_T|\Pc^v_T|}\Bigr)
\geq\tfrac{\sum_{v,u}\rho_uz^{*v}_u}{3\sum_vz^{*v}_v}=\OPT/3. \qedhere
\end{equation}
\end{proofnobox}

\begin{remark}
The above algorithm and analysis also show that the weaker LP where we replace the
constraints $x^v\bigl(\dt^\out(r)\bigr)=z^v_v$ for all $v\in V$, $\sum_v z^v_v=1$ with
$\sum_vx^v\bigl(\dt^\out(r)\bigr)=1$,\ $z^v_u\leq z^v_v$ for all $u,v\in V$, also has
integrality gap at most $3$. 
\end{remark}

\paragraph{Regret orienteering.}
The following variant of rooted orienteering, which we call {\em regret orienteering}, 
will be useful in Section~\ref{sec:ptp_orient}. In regret orienteering, 
instead of a cost bound $B$, we are given a {\em regret bound} 
$R$, and we seek a rooted path of {\em regret} at most $R$ that collects the maximum
reward. The LP-relaxation for regret-orienteering is very similar to \eqref{lp:root_or}; 
the only changes are that $z^v_v$ now indicates if $v$ is the {\em end node} of the
optimum path, and so we drop \eqref{pref2} and replace \eqref{cbound} with 
$\sum_{a\in A}c_ax^v_a\leq (\dist_v+R)z^v_v$. The rounding algorithm is essentially
unchanged: we convert the trees obtained from $x^v$ into $r$-$v$ paths, which are then
split into paths of regret at most $R$. Theorem~\ref{thm:rt_gap} yields the following
corollary.

\begin{corollary} \label{regorthm}
There is an LP-based $3$-approximation for regret orienteering.
\end{corollary}

\paragraph{Integrality gaps for weaker LPs.}
Recall that in sketching our rounding algorithm, we assumed at first (for simplicity) that
we know the node $v$ on the optimum path that has maximum distance from the root, and
impose in our LP that $x$ is a preflow under which the $r\leadsto v$ connectivity is $1$,
and that $x$ only visits nodes $u$ with $\dist_u\leq\dist_v$. 
We conclude this section by demonstrating that it is {\em crucial} to impose {\em both}
these constraints. 
First, consider the following LP relaxation that simply encodes that $x$ is a preflow of
value 1 but does not require that the $r\leadsto v$ connectivity under $x$ is 1 for 
any specific $v \in V$. 
As before, we have variables $z_u$ for each $u \in V$ denoting the $r\leadsto u$
connectivity under $x$.  
\begin{alignat}{3}  
\max & \quad & \sum_{u \in V}\rho(u)z_u & \tag{{R-O2}} \label{lp:root_or_2} \\
\text{s.t.} && 
x(\dt^\into(u)\bigr) & \geq x\bigl(\dt^\out(u)\bigr) \qquad && 
\forall u\in V \notag \\
&& x\bigl(\delta^{\into}(S)\bigr) & \geq z_u \qquad 
&& \forall S\subseteq V, u\in S \notag \\
&& \sum_{a\in A}c_ax_a & \leq B  && \notag \\
&& x\bigl(\dt^\out(r)\bigr) & = 1 \quad \notag \\
&& \quad x,z \geq 0. \notag
\end{alignat}
The above LP has unbounded integrality gap. 
Consider the following instance where the budget $B \geq 1$ is an integer. 
We have the following metric over $V \cup \{r\}$ where 
$V = \{r', v_1, v_2, \ldots, v_B\}$: 
each $v_i$ is at distance $B$ from $r$, and the distance between any two distinct 
$v_i, v_j$ is 1; $r'$ is at distance $1$ from $r$ and distance $B-1$ from every $v_i$. 
Let $\rho(v)=1$ for all $v\in V$.
An $r$-rooted path of length $B$ may visit $r'$ and at most one other node in $V$, so the
optimum solution has value 2. 

Consider the preflow $x$ that sends $\frac{1}{2}$ unit of flow along the path
$P_1=r\rightarrow r'\rightarrow v_1\rightarrow v_2\rightarrow \ldots\rightarrow v_B$, and
$\frac{1}{2}$ unit of flow along the path $P_2=r\rightarrow r'$. Letting
$z_v=\frac{1}{2}$ for all  $v\in V\sm\{r'\}$ and $z_{r'}=1$, it is easy to verify that
$(x,z)$ is a feasible solution to \eqref{lp:root_or_2} and has objective value
$\frac{B}{2}+1$. Thus, the integrality gap of \eqref{lp:root_or_2} is unbounded. 
Now consider the modification of \eqref{lp:root_or_2}, where we select some node $v$
and impost that the $r\leadsto v$ connectivity under $x$ is 1 (i.e., we add the
constraint $z_v=1$), but do not require that $x$ only visits nodes $u$ with
$\dist_u\leq\dist_v$. This LP continues to have an unbounded integrality gap, since if   
$v=r'$, the above $(x,z)$ continues to be a feasible solution to this LP. 

\medskip
In the context of LP \eqref{lp:root_or}---where we avoid the need for guessing the
maximum-distance node on the optimum path by having a separate collection 
$\{x^v_a\}, \{z^v_u\}$ of variables for all $v\in V$---the above example shows 
that it is important to impose constraints \eqref{pref2} {\em and} \eqref{cbound}.
In the absence of constraint \eqref{pref2}, letting $x^{r'}$ be the preflow that sends   
$\frac{1}{2}$ unit of flow each along $P_1$ and $P_2$, and setting $z^{r'}_{r'}=1$,
$z^{r'}_v=\frac{1}{2}$ for all $v\in V\sm\{r'\}$, yields a feasible solution to the
resulting LP (of value $\frac{B}{2}+1$). 
On the other hand, retaining \eqref{pref2}, but {\em weakening} \eqref{cbound} to
$\sum_{v,a}c_ax^v_a\leq B$ also yields an unbounded integrality gap: letting $x^{v_B}$ 
be the preflow that sends $\frac{1}{2}$ unit of flow along $P_1$, and $x^{r'}$ be the
preflow that sends $\frac{1}{2}$ unit of flow along $P_2$ (and setting the $z^v_u$s
appropriately) yields a feasible solution to the resulting LP (of value
$\frac{B}{2}+1$).

\section{Point-to-point orienteering} \label{sec:ptp_orient}
We now consider the generalization of rooted orienteering, where we have a start
node $r$ {\em and} an end node $t$, and we seek an $r$-$t$ path with cost at most $B$ that
collects the maximum reward. We may assume that $r$ and $t$ have 
$0$ reward, i.e., $\rho(r)=\rho(t)=0$.
The main result of this section is a novel reduction showing that 
{\em point-to-point (\ptp) orienteering} problem can be {\em reduced} to regret
orienteering losing a factor of at most 2 (Theorem~\ref{p2predn}). 
Combining this with our LP-approach for regret orienteering and Corollary~\ref{regorthm},
we obtain an LP-relaxation for \ptp-orienteering having integrality gap
at most 6 (Section~\ref{append-p2plp}). 
We believe that the insights gained from this reduction may find further application.  

\begin{theorem} \label{p2predn}
An $\al$-approximation algorithm for regret orienteering (where $\al\geq 1$) can be used
to obtain a $2\al$-approximation algorithm for \ptp-orienteering.
\end{theorem}

\begin{proofbox}
Let $\bigl(G=(\{r,t\}\cup V,E),\{c_{uv}\},\{\rho(u)\},B\bigr)$ be an instance of
\ptp-orienteering. 
Our reduction is simple. Let $P^*$ be an optimal solution. 
We ``guess'' a node $v\in P^*$ (which could be $r$ or $t$) such that 
$\dist_v+c_{vt}=\max_{u\in P^*}(\dist_u+c_{ut})$. (That is, we enumerate over all choices
for $v$.) Let $S=\{u\in \{r,t\}\cup V:\dist_u+c_{ut}\leq\dist_v+c_{vt}\}$.
We then consider two regret orienteering problems, both of which have regret bound 
$R=B-\dist_v-c_{vt}$ and involve only nodes in $S$ (i.e., we equivalently set $\rho(u)=0$
for all $u\notin S$); the first problem has root $r$, and the second has root $t$.
Let $P_1$ and $P_2$ be the solutions obtained for these two problems respectively by our 
$\al$-approximation algorithm. So for some $u_1,u_2\in S$, $P_1$ is an $r$-$u_1$ path, and
$P_2$ may be viewed as a $u_2$-$t$ path. 
Notice that $P_1$ appended with the edge $u_1t$ yields an $r$-$t$ path of {\em cost} 
at most $\dist_{u_1}+c^\reg(P_1)+c_{u_1t}\leq\dist_{u_1}+c_{u_1t}+B-\dist_v-c_{vt}\leq B$, 
since $u_1\in S$.
Similarly $P_2$ appended with the edge $ru_2$ yields an $r$-$t$ path of cost at most
$B$. We return $P_1+u_1t$ or $ru_2+P_2$, whichever has higher reward.

To analyze this, we observe that the $r$-$v$ portion of $P^*$ is a feasible solution to
the regret-orienteering instance with root $r$, since its cost is at most $B-c_{vt}$, and
hence, its regret is at most $R$. Similarly, the $v$-$t$ portion of $P^*$ (viewed in
reverse) is a feasible solution to the regret-orienteering instance with root
$t$. Therefore, $\max\bigl\{\rho(P_1+u_1t),\rho(ru_2+P_2)\}\geq\rho(P^*)/2\al$.
\end{proofbox}

\subsection{LP-relaxation for \ptp-orienteering and rounding algorithm} \label{append-p2plp}
As in the case of rooted orienteering, we replace the ``guessing'' step by having an
indicator variables $z_v^v$ to denote if $v$ is the node with maximum $\dist_v+c_{vt}$ on
the optimum path. 
As suggested by the proof of Theorem~\ref{p2predn}, our LP then incorporates ideas from
the rooted-orienteering LP \eqref{lp:root_or} to encode that, our solution is, to an
extent of $z_v^v$, a combination of $r$-$v$ and $v$-$t$ paths of regret at most
$B-\dist_v-c_{vt}$, with respect to roots $r$ and $t$ respectively, that only visit nodes
$u$ with $\dist_u+c_{ut}\leq\dist_v+c_{vt}$.  
(We get some notational savings since we do not need to ``guess'' the endpoints of the
desired regret-$R$ paths: $v$ is the endpoint of both paths, and so we may work with
$r$-$v$ and  $v$-$t$ flows, instead of preflows.) 

Let $D=(V':=\{r,t\}\cup V,A)$ denote the bidirected version of $G$. (Again, both $(u,v)$
and $(v,u)$ get cost $c_{uv}$.)
For every $v\in V'$, 
we let $x^{rv}$ denote an $r$-$v$ flow of value $z^v_v$, and ${x}^{vt}$
denote a $v$-$t$ flow of value $z^v_v$. We impose that
$x^{rv}\bigl(\dt^\into(u)\bigr)=x^{vt}\bigl(\dt^\into(u)\bigr)=0$ 
whenever $\dist_u+c_{ut} > \dist_v + c_{vt}$.  
We use $z^{rv}_u$ and $z^{vt}_u$ to denote respectively the $r\leadsto u$ connectivity
under $x^{rv}$ and the $v\leadsto u$ connectivity under $x^{vt}$. So in an integral
solution. $z^{rv}_u$ and ${z}^{vt}_u$ indicate respectively if $u$ lies on the $r$-$v$
portion or on the $v$-$t$ portion of the optimum path.  
For nodes $v,p,q \in V'$ and $\kappa \geq 0$, define 
\begin{align*}
\mathcal F_v(p,q,\kappa):=\biggl\{
x\in\R_+^A: \ \ x\bigl(\delta^{\out}(p)\bigr)&=\kp=
x\bigl(\delta^{\into}(q)\bigr), \qquad 
x\bigl(\delta^{\into}(p)\bigr)=0=x\bigl(\delta^{\out}(q)\bigr) \\
x\bigl(\delta^{\into}(w)\bigr)&-x\bigl(\delta^{\out}(w)\bigr)=0 
\qquad \forall w\in V'\sm\{p,q\} \\
x\bigl(\dt^\into(w)\bigr)&=0 \qquad \forall w\in V': \dist_w+c_{wt}>\dist_v+c_{vt}
\biggr\}
\end{align*}
Note that if $\kp>0$, then $\F(u,u,\kp)=\es$ for every $u$. Recall that
$\rho(r)=\rho(t)=0$. 
\begin{alignat}{3}  
\max & \ \ & \sum_{u,v\in V'}\rho(u)(z^{rv}_u+z^{vt}_u) \tag{{P2P-O}} \label{lp:p2p_or} \\
\text{s.t.} & \quad & x^{rv}\!\in\!\F_v(r,v,z^v_v), \ x^{vt} & \!\in\!\F_v(v,t,z^v_v) 
\quad \ && \forall v\in V' \notag \\
&& x^{rv}\bigl(\delta^{\into}(S)\bigr) & \geq z^{rv}_u \qquad && 
\forall v\in V', S\subseteq V'\sm\{r\}, u\in S \label{rvcon} \\
&& x^{vt}\bigl(\delta^{\into}(S)\bigr) & \geq z^{vt}_u \qquad && 
\forall v\in V', S\subseteq V'\sm\{v\}, u\in S \label{vtcon} \\
&& \sum_{a\in A}c_a(x^{rv}_a+x^{vt}_a) & \leq B z^v_v  && \forall v \in V' \label{vbound} \\
&& \sum_{v\in V'} z^v_v & = 1, && x,z \geq 0. \notag
\end{alignat}
As before, we can model the cut constraints \eqref{rvcon}, \eqref{vtcon} using
additional flow variables and constraints to obtain a compact formulation.

As with \eqref{lp:root_or}, the constraints of \eqref{lp:p2p_or} imply that
$z^{rv}_u=z^{vt}_u=0$ if $\dist_u+c_{ut}>\dist_v+c_{vt}$, and $z^{rv}_u,z^{vt}_u\leq z^v_v$
for all $u$.
We remark that we could further add the constraint $z^{rv}_u + {z}^{vt}_u \leq z^{v}_v$
for $u,v \in V'$ to obtain a stronger relaxation whose integer solutions correspond to
feasible point-to-point orienteering solutions of the same cost. 
But we omit this since we do not need it in our rounding procedure.
Let $(x^*,z^*)$ be an optimal solution to \eqref{lp:p2p_or} and $\OPT$ be its value. 

\begin{theorem}\label{thm:ptp_gap}
We can round $(x^*,z^*)$ to a solution to \ptp-orienteering of value at least $\OPT/6$.
\end{theorem}

\begin{proofnobox}
Consider some $v \in V'$. Let $R_v=B-\dist_v-c_{vt}$ be the regret bound that we require
for the $r$-$v$ and $v$-$t$ paths with respect to roots $r$ and $t$ respectively.
Let $S_v=\{u\in V':\dist_u+c_{ut}\leq\dist_v+c_{vt}\}$ be the set of nodes that these
paths may visit.
It is easy to see via flow decomposition that $c^Tx^{rv}\geq\dist_vz^v_v$ and
$c^Tx^{vt}\geq c_{vt}z^v_v$. Therefore, \eqref{vbound} implies that
$c^Tx^{rv}\leq(\dist_v+R_v)z^v_v$ and $c^Tx^{vt}\leq(c_{vt}+R_v)z^v_v$.

We can apply the rounding procedure in the proof of Theorem~\ref{thm:rt_gap} to the flow
$x^{rv}$ to obtain an $r$-rooted path ending at some node $w\in S_v$ having regret at most
$R_v$, and gathering reward at least $\bigl(\sum_{u\in V'}\rho(u)z^{rv}_u\bigr)/3z^v_v$
(see \eqref{orineq2}, \eqref{orineq3}). Appending the edge $wt$ to this path, we obtain an
$r$-$t$ path, which we denote $Q^{rv}$, of $c$-cost at most $R_v+\dist_w+c_{wt}\leq B$.

We can apply the same process to the $x^{vt}$ flow. 
In particular, let $\ox^{tv}_{uw} = {x}^{vt}_{wu}$. Then, $\ox^{tv}$ is a $t$-$v$ flow of
value ${z}^{v}_v$ with $c$-cost at most $(c_{vt}+R_v){z}^{v}_v$ (as $c$ is symmetric) and 
gathering reward $\sum_{u\in V}\rho(u)z^{vt}_u$. 
Therefore, we find a $t$-rooted path ending at some node $w\in S_v$ having regret at most
$R_v$ with respect to $t$ and gathering reward at least 
$\bigl(\sum_{u\in V'}\rho(u)z^{vt}_u\bigr)/3z^v_v$. Viewing this path as a $w$-$t$ path and
appending the edge $rw$, we obtain an $r$-$t$ path, which we denote $Q^{vt}$, of $c$-cost
at most $R_v+c_{wt}+\dist_w\leq B$. 

We return the maximum-reward path among the collection $\{Q^{rv}, Q^{vt}: v \in V'\}$.
The reward of this path is at least
\begin{equation}
\frac{\sum_{v\in V'}z^v_v\bigl(\rho(Q^{rv})+\rho(Q^{vt})\bigr)}{2\sum_{v\in V'}z^v_v}
\geq\frac{\sum_{v,u\in V'}\rho(u)\bigl(z^{rv}_u+z^{vt}_u\bigr)}{6}=\OPT/6. \qedhere
\end{equation}
\end{proofnobox}

\section{Compact LPs and improved guarantees for \rvrp} 
\label{sec:rvrp}
Recall that in the {\em regret-bounded vehicle routing problem} (\rvrp), we are given an
undirected complete graph $G=(\{r\}\cup V,E)$ on $n$ nodes with a distinguished root
(depot) node $r$, metric edge costs or distances $\{c_{uv}\}$, and a regret-bound $R$.  
The goal is to find the minimum number of rooted paths that cover all nodes so that the
regret of each node with respect to the path covering it is at most $R$. Throughout, let
$\iopt$ denote the optimal value of the \rvrp instance.
We describe two compact LP-relaxations for \rvrp and corresponding
rounding algorithms that yield improvements, in both approximation ratio and running time,    
over the \rvrp-algorithm in~\cite{FriggstadS14}.   
In Section~\ref{orientreglp}, we observe that the compact LP for orienteering
\eqref{lp:root_or} yields a natural LP for \rvrp; by combining the rounding ideas 
used for orienteering and Theorem~\ref{fs-rvrpthm},
we obtain a $27$-approximation algorithm for \rvrp.
In Section~\ref{newreglp}, we formulate an unorthodox, stronger LP-relaxation \eqref{rlp2}
for \rvrp by leveraging some key structural insights in~\cite{FriggstadS14}. We devise a
rounding algorithm for this LP that leads to a 15-approximation algorithm for \rvrp, which
is a significant improvement over the guarantee obtained in~\cite{FriggstadS14}.

\subsection{Extending the orienteering LP to \rvrp} \label{orientreglp}
The LP-relaxation below can be viewed as a natural variant of the orienteering LP
adapted to \rvrp. As before, let $D=(\{r\}\cup V,A)$ be the bidirected version of
$G$. For each node $v$, 
$x^v$ is a preflow (constraint \eqref{preflow}) of value $z^v_v$ such that the
$r\leadsto u$ connectivity under capacities $\{x^v_a\}$ is 
{at least $z^v_u$ for all $u, v$ (constraint \eqref{ruvcon}).}
\begin{alignat}{3}
\min & \quad & \sum_v z^v_v & \tag{R1} \label{rlp1} \\
\text{s.t.} && \quad x^v(\dt^\into(u)\bigr) & \geq x^v\bigl(\dt^\out(u)\bigr) \qquad && 
\forall u,v\in V \label{preflow} \\
&& x^v\bigl(\delta^{\into}(S)\bigr) & \geq z^v_u \qquad && \forall v\in V, 
S \subseteq V, u\in S \label{ruvcon} \\
&& \sum_{a\in A} c_ax^v_a & \leq (\dist_v+R)z^v_v \qquad && \forall v\in V \notag \\
&& x^v\bigl(\dt^\out(r)\bigr) & = z^v_v \quad \forall v\in V, \qquad 
&& \sum_{v\in V} z^v_u \geq 1 \quad \forall u\in V, \qquad
x,z \geq 0. \notag
\end{alignat}
As before, we can obtain a compact formulation by replacing the cut constraints
\eqref{ruvcon} with constraints involving suitable flow variables.
Let $(x^*,z^*)$ be an optimal
solution to \eqref{rlp1}, and $\OPT$ be its objective value. Note that
$\ceil{\OPT}\leq\iopt$. 

\begin{theorem} \label{rlp1thm}
We can round $(x^*,z^*)$ to obtain a $27$-approximation for \rvrp.
\end{theorem}

\begin{proof} 
Apply Theorem~\ref{arbpack} to each preflow $x^{*v}$ taking $K=z_v^v$, to decompose
$x^{*v}$ into $r$-rooted out-arborescences, which we view as rooted trees. This yields
nonnegative weights $\{\gm^v_T\}_{T\in\T}$ such that $\sum_T\gm^v_T=z^v_v$, 
$\sum_T\gm^v_Tc(T)=\sum_a c_ax^v_a\leq(\dist_v+R)z^v_v$, and
$\sum_{T:u\in T}\gm^v_T\geq z^u_v$ for all $u\in V$. 
Note that $v\in T$ whenever $\gm^v_T>0$.  
Doubling the edges not lying on the $r$-$v$ paths of these trees and shortcutting, gives a 
collection $\Pc_v$ of simple $r$-$v$ paths having total regret at most $2R\cdot z_v^v$. 
Thus, $\bigcup_{v\in V}\Pc_v$ is a collection of rooted paths covering
each $u \in V$ to an extent of 1 and having total regret cost at most $2R\cdot\OPT$.
Applying Theorem~\ref{fs-rvrpthm} with $\dt=\frac{1}{3}$ to this collection yields
an \rvrp solution with at most $24\cdot\OPT+\ceil{3\cdot\OPT}\leq
24\cdot\OPT+3\ceil{\OPT}\leq 27\cdot\iopt$ paths. 
\end{proof}

\subsection{A new compact LP for \rvrp leading to a 15-approximation} \label{newreglp}
We now propose a different 
LP for \rvrp, which leads to a much-improved
15-approximation for \rvrp. To motivate this LP, we first collect some facts
from~\cite{FriggstadS14,BlumCKLMM07} pertaining to the regret objective.   
By merging all nodes at distance 0 from each other, we may assume that $c_{uv}>0$ for all
$u,v\in V\cup\{r\}$, and hence $\dist_v>0$ for all $v\in V$.

\begin{defn}[\cite{FriggstadS14}] \label{redblue}
Let $P$ be a rooted path ending at $w$. Consider an edge $(u,v)$ of $P$, where $u$ 
precedes $v$ on $P$. We call this a {\em red} edge of $P$ if there exist nodes $x$ and $y$  
on the $r$-$u$ portion and $v$-$w$ portion of $P$ respectively such that
$\dist_x\geq\dist_y$; otherwise, we call this a {\em blue} edge of $P$. 
For a node $x\in P$, let $\red(x,P)$ denote the maximal subpath $Q$ of $P$ containing $x$ 
consisting of only red edges (which might be the trivial path $\{x\}$).
\end{defn}

Note that the first edge of a rooted path $P$ is always a blue edge.
Call the collection 

\begin{lemma}[\cite{BlumCKLMM07}] \label{lem:redcost} 
For any rooted path $P$, we have
$\sum_{\text{$e$ red on $P$}} c_e \leq \frac{3}{2} c^{\reg}(P)$.
\end{lemma}

\begin{lemma}[\cite{FriggstadS14}] \label{lem:cross} 
(i) Let $u, v$ be nodes on a rooted path $P$ such that $u$ precedes $v$ on $P$ and
$\red(u,P) \neq \red(v,P)$; then $\dist_{u} < \dist_{v}$.
(ii) If $P'$ is obtained by shortcutting $P$ so that it contains at most one node from
each red interval of $P$, then for every edge $(x,y)$ of $P'$ with $x$ preceding $y$ on
$P'$, we have $\dist_x < \dist_y$.
\end{lemma}

We say that a node $u$ on a rooted path of $P$ is a {\em sentinel} of $P$ if $u$ is the
first node of $\red(u,P)$. Part (ii) above shows that if we shortcut each path $P$ of an
optimal \rvrp-solution past the non-sentinel nodes of $P$, then we obtain a
distance-increasing collection of paths. Moreover, part (i) implies that if $x$ and $y$
are sentinels on $P$ with $x$ appearing before $y$, then 
$\max_{u\in\red(x,P)}\dist_u<\min_{u\in\red(y,P)}\dist_u$.
Finally, every non-sentinel node is connected to the sentinel corresponding to its red
interval via red edges, and Lemma~\ref{lem:redcost} shows that the total ($c$-) cost
of these edges at most $1.5R(\text{optimal value})$. 

Thus, we can view an \rvrp-solution as a collection of distance-increasing rooted paths
covering some sentinel nodes $S$, and a low-cost way of connecting the nodes in $V\sm S$ 
to $S$. Our LP-relaxation searches for the best such solution. 
Let $\D:=\{\dist_v:v\in V\}$. For every $u\in V$, define $\D_u$ to be the collection
$\bigl\{[d_1,d_2]: d_1,d_2\in\D,\ d_1\leq\dist_u\leq d_2\bigr\}$ of (closed) intervals. 
We have variables $x_{u,I,u}$ for every node $u\in V$ and interval $I=[d_1,d_2]\in\D_u$ 
to indicate if $u$ is a sentinel and $d_1, d_2$ are the minimum and maximum
distances (from $r$) respectively of nodes in the red interval corresponding to $u$; we 
say that $I$ is $u$'s distance interval. 
We also have variables $x_{u,I,v}$ for $v\neq u$ to indicate that $v$ is connected to
sentinel $u$ with distance interval $I$, and edge variables $\{z_e\}_{e\in E}$ that
encode these connections. 
Finally, we have flow variables $f_{r,u,I}, f_{u,I,v,J}, f_{u,I,t}$ for all $u,v\in V$ and 
$I\in\D_u$, $J\in\D_v$ that encode the distance-increasing rooted paths on the
sentinels, with $t$ representing a fictitious sink. We include constraints that encode
that the distance intervals of sentinels lying on the same path are disjoint, and a 
non-sentinel $v$ can be connected to $(u,I)$ only if $\dist_v\in I$.
{We obtain the following LP.}

\begin{alignat}{3}
\min & \quad & \sum_{u\in V,I\in\D_u} f_{r,u,I} \qquad & \tag{R2} \label{rlp2} \\
\text{s.t.} && \sum_{u\in V,I\in\D_u}x_{u,I,v} & \geq 1 \qquad &&
\forall v\in V \label{vcov} \\
&& x_{u,I,v}\leq x_{u,I,u}, \quad x_{u,I,v} & = 0 \quad \text{if $\dist_v\notin I$} \qquad 
&& \forall u,v\in V, I\in\D_u \label{vasgn} \\
&& z\bigl(\delta(S)\bigr) & \geq \sum_{u\notin S,I\in\D_u}x_{u,I,v}  
\qquad && \forall v\in V, \{v\}\sse S\sse V \label{vcon} \\
&& f_{r,u,I}+\sum_{v\in V,J\in\D_v}f_{v,J,u,I} & = x_{u,I,u} 
\qquad && \forall u\in V, I\in\D_u \label{usentin} \\
&& \sum_{v\in V,J\in\D_v}f_{u,I,v,J}+f_{u,I,t} & = x_{u,I,u} 
\qquad && \forall u\in V, I\in\D_u \label{usentout} \\
&& f_{u,I,v,J} & = 0 \qquad && 
\begin{aligned} & \forall u,v\in V, I\in\D_u, J\in\D_v: \\
& I\cap J\neq\es\text{ or }\dist_v\leq\dist_u \end{aligned} \label{dinc} \\
&& \sum_{u,v\in V,I\in\D_u,J\in\D_v}c^\reg_{uv}f_{u,I,v,J} & \leq 
R\cdot\sum_{u\in V,I\in\D_u}f_{r,u,I} \label{costf} \\
&& \sum_{e\in E}c_ez_e & \leq 1.5R\cdot\sum_{u\in V,I\in\D_u}f_{r,u,I} \label{costz} \\
&& x,z,f & \geq 0. \notag
\end{alignat} 
Constraint \eqref{vcov} encodes that every node $v$ is either a sentinel or is
connected to a sentinel; \eqref{vasgn} ensures that if $v$ is assigned to $(u,I)$, then
$u$ is indeed a sentinel with distance interval $I$ and that $\dist_v\in I$. Constraints
\eqref{vcon} ensure that the $z_e$s (fractionally) connect each non-sentinel $v$ to the
sentinel specified by the $x_{u,I,v}$ variables.
Constraints \eqref{usentin}, \eqref{usentout} encode that each sentinel $(u,I)$ lies on
rooted paths, and 
\eqref{dinc} ensures that these paths are distance increasing and moreover the distance
intervals of the sentinels on the paths are disjoint. Finally, letting $k$ denote the
number of paths used, \eqref{costf}, \eqref{costz} encode that the total regret of the
distance-increasing paths is at most $kR$ (note that $c^\reg_{ru}=0$ for all $u$), and the
total cost of the edges used to connect non-sentinels to sentinels is at most $1.5kR$. 
As before, the cut constraints \eqref{vcon} can be equivalently stated using flows to
obtain a polynomial-size LP.
Let $(x^*,z^*,f^*)$ be an optimal solution to \eqref{rlp2} and
$\OPT$ denote its objective value. We have already argued that
an optimal \rvrp-solution yields an integer solution to \eqref{rlp2}, so we obtain the
following. 

\begin{lemma} \label{rlp2relax}
$\ceil{\OPT}$ is at most the optimal value, $\iopt$, of the \rvrp instance.
\end{lemma}

We remark that an integer solution to \eqref{rlp2} {\em need not} correspond to an \rvrp
solution since constraints \eqref{vcon} only ensure that non-sentinels are connected to
sentinels, but not necessarily via paths. 
Nevertheless, we show that we can round $(x^*,z^*,f^*)$ to an \rvrp-solution using at most
$15\cdot\ceil{\OPT}$ paths.  

Our rounding algorithm proceeds in a similar fashion as the \rvrp-algorithm
in~\cite{FriggstadS14}; yet, we obtain an improved approximation ratio since one can solve
\eqref{rlp2} exactly whereas one can only obtain a $(2+\e)$-approximate solution to the
configuration LP in~\cite{FriggstadS14}. Let $\tht\in(0,1)$ be a parameter that we will set
later. We first obtain a forest of $c$-cost at most $\frac{3R}{1-\tht}\cdot\OPT$ such that
every component $Z$ contains a witness node $v$ that is assigned to an extent of at least
$\tht$ to sentinels in $Z$. We argue that if we contract the components of $F$, then the
distance-increasing sentinel flow paths yield an acyclic flow that covers every contracted
component to an extent of at least $\tht$. Hence, using the integrality property of flows,
we obtain an integral flow, and hence a collection of at most $\ceil{\frac{\OPT}{\tht}}$
rooted paths, that covers every component and has cost at most
$\frac{R}{\tht}\cdot\OPT$. Next, we show that we can uncontract the components and attach
the component-nodes to these rooted paths incurring an additional cost of at most
$\frac{6R}{1-\tht}\cdot\OPT$. Finally, by applying Lemma~\ref{avg2max}, we obtain an \rvrp
solution with at most
$\bigl(\frac{6}{1-\tht}+\frac{1}{\tht}\bigr)\OPT+\ceil{\frac{\OPT}{\tht}}$ rooted paths. We
now describe the algorithm in detail and analyze it.

{\small \vspace{5pt}
\hrule \vspace{-2pt} 
\begin{enumerate}[label=A\arabic*., itemsep=0.5ex]
\item For $S\sse V$, define $h(S)=1$ if $\sum_{u\in S,I\in\D_u}x^*_{u,I,v}<\tht$ for all
$v\in S$, and $0$ otherwise. 
$h$ is a {\em downwards-monotone} cut-requirement
function: if $\es\neq A\sse B$, then $h(A)\geq h(B)$. 
Use the LP-relative 2-approximation algorithm in~\cite{GoemansW94} for $\{0,1\}$
downwards-monotone functions to obtain a forest $F$ such that 
{$|\dt(S)\cap F|\geq h(S)$ for all $S\sse V$.} 

\item For every component $Z$ of $F$ with $r\notin Z$, pick a {\em witness node} $w\in Z$
such that $\sum_{u\in Z,I\in\D_u}x^*_{u,I,w}\geq \tht$. 
Let $\sg(w)=\{(u,I): u\in Z, x^*_{u,I,w}>0\}$.
Let $W\sse V$ be the set of all such witness nodes. 

\item $f^*$ is an $r\leadsto t$ flow in an auxiliary graph having nodes
$r$, $t$, and $(u,I)$ for all $u\in V, I\in\D_u$, edges $(r,(u,I))$, $((u,I),t)$ for 
all $u\in V, I\in\D_u$, and edges $((u,I),(v,J))$ for all $u,v\in V, I\in\D_u, J\in\D_v$
such that $\dist_u<\dist_v$ and $I\cap J=\es$. 
Let $\{f^*_P\}_{P\in\Pc}$ be a path-decomposition of this flow. 
Modify each flow path $P\in\Pc$ as follows.
First, drop $t$ from $P$.
Shortcut $P$ past the nodes in $P$ that are not in $\{r\}\cup\bigcup_{w\in W}\sg(w)$. 
The resulting path maps naturally to a rooted path in $G$ (obtained by simply dropping the
distance intervals), which we denote by $\pi(P)$. Clearly, 
$c^{\reg}\bigl(\pi(P)\bigr)\leq\sum_{((u,I),(v,J))\in P}c^\reg_{uv}$ since shortcutting
does not increase the regret cost.

\item Let $\Qc$ be the collection of rooted paths obtained by taking the paths 
$\{\pi(P): P\in\Pc\}$ and contracting the components of $F$. Let $H$ be the directed graph
obtained by directing the paths in $\Qc$ away from $r$. 
To avoid notational clutter, for a component $Z$ of $F$, we use $Z$ to also denote the
corresponding contracted node in $H$. 
For each $Q\in\Qc$, define $y_Q=\sum_{P\in\Pc:\pi(P)\text{ maps to }Q}f^*_P$. 
Lemma~\ref{acycflow} proves that $H$ is acyclic and $\sum_{Q\in\Qc: Z\in Q}y_Q\geq\tht$
for every component $Z$ of $F$. 

\item Use the integrality property of flows to round the flow
$\bigl\{\frac{y_Q}{\tht}\bigr\}_{Q\in\Qc}$ to 
an integer flow of value $k\leq\ceil{\frac{\OPT}{\tht}}$ and regret-cost at most
$\frac{R}{\tht}\cdot\OPT$. Since $H$ is acyclic, this yields rooted paths
$\hP_1,\ldots,\hP_k$ so that every component $Z$ of $F$ 
{lies on exactly one $\hP_i$ path.}

\item 
We map the $\hP_i$s to rooted paths in $G$ that cover $V$ as follows. 
Consider a path $\hP_i$. Let $Z$ be a component lying on $\hP_i$, and $u,v\in Z$ 
be the nodes where $\hP_i$ enters and leaves $Z$ respectively. We add to $\hP_i$ a
$u$-$v$ path that covers all nodes of $Z$ obtained by doubling all edges of $Z$
except those on the $u$-$v$ path in $Z$ and shortcutting. Let $\tP_i$ be the rooted path
in $G$ obtained by doing this for all components lying on $\hP_i$.

\item Finally, we use Lemma~\ref{avg2max} to convert $\tP_1,\ldots,\tP_k$ to an
\rvrp-solution. 
\end{enumerate} 
\vspace{-2pt}
\hrule}

\paragraph{Analysis.} 
We first bound the cost of the forest $F$ obtained in step A1 in Lemma~\ref{forest}, which
also yields a bound on the additional cost incurred in step A6 to convert the $\hP_i$s to
the rooted paths $\tP_i$s. Lemma~\ref{acycflow} proves that $H$ is acyclic, and that $y$
covers each component of $F$ to an extent of at least $\tht$. Theorem~\ref{rlp2thm}
combines these ingredients to obtain the stated performance guarantee.

\begin{lemma} \label{forest}
The forest $F$ obtained in step A1 satisfies $c(F)\leq\frac{3R}{1-\tht}\cdot\OPT$. 
\end{lemma}

\begin{proof} 
If $h(S)=1$, then we have 
$z^*\bigl(\dt(S)\bigr)\geq\sum_{u\notin S,I\in\D_u}x^*_{u,I,v}\geq 1-\tht$ for all $v\in S$
due to \eqref{vcov}, \eqref{vcon}. 
So the forest $F$ has cost at most $2\bigl(\sum_e c_ez^*_e\bigr)/(1-\tht)$ which is at most
$\frac{3R}{1-\tht}\cdot\OPT$ (due to \eqref{costz}).
\end{proof}

\begin{lemma} \label{helper}
Let $P\in\Pc$ and $Z$ be a component of $F$. Let $w$ be the witness node of $Z$.
Then the number of nodes in $\pi(P)\cap Z$ is equal to $|P\cap\sg(w)|\leq 1$.
\end{lemma}

\begin{proofbox} 
Suppose $u$ be a node in $\pi(P)\cap Z$. Then, there is a unique $I\in\D_u$ such that
$(u,I)\in P$. It must be that $(u,I)\in\sg(w)$, as otherwise, since $(u,I)$ cannot be in
$\sg(w')$ for any other witness node $w'$, we would have shortcut $P$ past
$(u,I)$. Conversely, if $(u,I)\in P\cap\sg(w)$, then by construction, we have $u\in\pi(P)$.
Thus, the number of nodes in $|\pi(P)\cap Z|$ is equal to $|P\cap\sg(w)|$. 

Finally, we have $|P\cap\sg(w)|\leq 1$ for any flow path $P\in\Pc$ and any $w\in W$, 
since if $(u,I)\in P\cap\sg(w)$ ($u$ could be $w$) then
$\dist_w\in I$ and the distance intervals corresponding to nodes on $P$ are disjoint. 
\end{proofbox}

\begin{lemma} \label{acycflow}
$H$ is acyclic. 
\mbox{For any component $Z$ of $F$, $\sum_{Q\in\Qc: Z\in Q}y_Q\geq\tht$.}
\end{lemma}

\begin{proof} 
Let $Z'$ be a node of $H$, and $w'$ be the witness node of $Z'$. Give $Z'$ the label
$\dist_{w'}$. We claim that sorting the nodes of $H$ in increasing order of their labels
yields a topological ordering of $H$, showing that $H$ is acyclic. 
Let $(Z_1,Z_2)$ be an arc of $H$. 
Let $w_1$, $w_2$ be the witness nodes corresponding to $Z_1$ and $Z_2$ respectively. Then
there is a flow path $P\in\Pc$ and some edge $(u_1,u_2)$ of $\sg(P)$ where $u_1\in Z_1$, 
$u_2\in Z_2$. So there exist $(u_1,I_1)\in P\cap\sg(w_1)$ and $(u_2,I_2)\in P\cap\sg(w_2)$
such that $(u_2,I_2)$ appears after $(u_1,I_1)$ on $P$. Then, $\dist_{u_1}<\dist_{u_2}$, 
$\dist_{u_1},\dist_{w_1}\in I_1$, $\dist_{u_2},\dist_{w_2}\in I_2$, and $I_1\cap I_2=\es$. 
This implies that $\dist_{w_1}<\dist_{w_2}$.

Let $w$ be the witness node of $Z$. We have 
$\sum_{Q\in\Qc: Z\in Q}y_Q=\sum_{P\in\Pc:\pi(P)\cap Z\neq\es}f^*_P=\sum_{P\in\Pc:P\cap\sg(w)\neq\es}f^*_P$
where the last equality follows from Lemma~\ref{helper}. 
Since $|P\cap\sg(w)|\leq 1$ for all $P\in\Pc$, we have
$$
\sum_{P\in\Pc:P\cap\sg(w)\neq\es}f^*_P=\sum_{(u,I)\in\sg(w)}\sum_{P\in\Pc:(u,I)\in P}f^*_P
=\sum_{(u,I)\in\sg(w)}x^*_{u,I,u}\geq\tht.
$$
The second equality above follows from \eqref{usentin} and since $\{f^*_P\}_{P\in\Pc}$ is a
path decomposition of $f^*$. 
\end{proof}

\begin{theorem} \label{rlp2thm}
The above algorithm returns an \rvrp-solution with at most 
$\bigl(\frac{6}{1-\tht}+\frac{1}{\tht}\bigr)\OPT+\ceil{\frac{\OPT}{\tht}}$ paths. 
Thus, taking $\tht=\frac{1}{3}$, we obtain at most $15\cdot\iopt$ paths.
\end{theorem}

\begin{proofbox}
The total regret of the paths $\tP_1,\ldots,\tP_k$ is at most
$\bigl(\frac{6}{1-\tht}+\frac{1}{\tht}\bigr)R\cdot\OPT$. This follows because the total
regret of $\hP_1,\ldots,\hP_k$ is at most $\frac{R}{\tht}\cdot\OPT$, and the regret-cost of
the path added for each component $Z$ is at most the regret-cost of the tour $Y$ obtained
by doubling all edges of $Z$, and $c^\reg(Y)=c(Y)=2c(Z)$. Combining this with
Lemma~\ref{forest} proves the claim.
So applying Lemma~\ref{avg2max} to $\tP_1,\ldots,\tP_k$ yields the stated bound, and for
$\tht=\frac{1}{3}$, this bound translates to 
$12\cdot\OPT+\ceil{3\cdot\OPT}\leq 15\cdot\iopt$.
\end{proofbox}

\section{Minimum-regret TSP-path} \label{minregtsp}
We now consider the {\em minimum-regret TSP-path} problem, wherein we have (as
before), a complete graph $G=(V',E)$, $r,t\in V'$, metric edge costs $\{c_{uv}\}$, and we seek a
minimum-regret $r$-$t$ path that visits all nodes.
Observe that this is precisely the \atsp{\em -path} problem
under the asymmetric regret metric $c^\reg$. We establish a tight bound of 2 on the
integrality gap of the standard \atsp-path LP for the class of regret-metrics (induced by
a symmetric metric). 
We consider the following LP for min-regret TSP path.
Let $D=(V',A)$ be the bidirected version of $G$. Let $b_t=1=-b_r$ and $b_v=0$ for all
$v\in V'\sm\{r,t\}$.
\begin{alignat}{1}
\min \ \ \sum_{a\in A} c^{\reg}_ax_a \quad \text{s.t.} \quad 
x\bigl(\delta^{\into}(v)\bigr)&-x\bigl(\delta^{\out}(v)\bigr)=b_v\ \ \forall v\in V',
\quad x\geq 0 \tag{{R-TSP}} \label{lp:atsp} \\
x(\delta^{\into}(S)) & \geq 1 \quad \forall \es\neq S\subseteq V\sm\{r\}. \label{xcon}
\end{alignat}
Clearly, \eqref{lp:atsp} is no stronger than the LP where we impose indegree and outdegree
constraints on the nodes. 
In contrast to Theorem~\ref{thm:atsp_gap}, for general asymmetric metrics, even this
stronger LP is only known to have 
integrality gap $O\left(\frac{\log n}{\log\log n}\right)$~\cite{AsadpourGMOS10,fgs:atspp}.
(The corresponding LP for \atsp has $\poly(\log\log n)$ integrality
gap~\cite{aog}.) 

\begin{theorem}\label{thm:atsp_gap}
The integrality gap of \eqref{lp:atsp} is 2 for regret metrics, 
and we can obtain an \atsp-path solution with $c^\reg$-cost at most
$2\cdot\OPT_{\text{\ref{lp:atsp}}}$ in polynomial time. 
\end{theorem}

\begin{proof} 
We first describe the rounding algorithm showing an integrality-gap upper bound of 2. 
It is convenient to consider the following weaker LP. 
\begin{equation}
\min \quad \sum_a c_ax_a-\dist_t \qquad \text{s.t.} \qquad
\text{$x$ is an $r$-preflow}, \quad \eqref{xcon}, \quad x\geq 0.
\tag{P} \label{watsp}
\end{equation}
LP \eqref{watsp} is weaker than \eqref{lp:atsp} because if $x$ is a feasible solution to
\eqref{lp:atsp} then it is clearly feasible to \eqref{watsp}, and
$$
\sum_{(u,v)\in A}c^\reg_{u,v}x_{u,v}=\sum_{a\in A}c_ax_a
+\sum_{u\in V'}\dist_u\Bigl(x\bigl(\dt^\out(u)\bigr)-x\bigl(\dt^\into(u)\bigr)\Bigr)
=\sum_{a\in A}c_ax_a-\dist_t.
$$ 
We show how to obtain an \atsp-path solution of $c^{\reg}$-cost at most
$2\OPT_{\text{\ref{watsp}}}$,
thereby showing that \eqref{watsp}, and hence \eqref{lp:atsp} has integrality gap at most
2. 

Let $x^*$ be an optimum solution to \eqref{watsp}. 
Since the $r\leadsto v$ connectivity is 1 under $x^*$, applying Theorem~\ref{arbpack} to
$x^*$ with $K=1$, yields a collection of $r$-rooted out-arborescences, all of which span
$V'$.  
As should be routine by now, we view these arborescences as spanning trees in $G$, convert
each spanning tree to an $r$-$t$ path via doubling and shortcutting, and return the path
with the smallest $c^\reg$-cost. The $c^\reg$-cost of the path obtained from tree $T$ is
at most $2\bigl(c(T)-\dist_T\bigr)$. The bound now follows because if $\{\gm_T\}$ are the
nonnegative weights obtained from Theorem~\ref{arbpack} (which sum up to 1), the
$c^\reg$-cost we obtain is at most
$\sum_T\gm_T\cdot 2\bigl(c(T)-\dist_t\bigr)=2\bigl(\sum_a c_ax^*_a-\dist_t\bigr)$.

As noted earlier, this integrality gap upper bound of $2$ can also be inferred from the
result of~\cite{ChaudhuriGRT03}, which (in particular) shows that one can obtain a
spanning tree of cost at most the min-cost Hamiltonian $r$-$t$ path.

\paragraph{Lower bound of 2 on the integrality gap.}
We show a lower bound of 2 on the integrality gap even for the stronger LP, where we
we additionally impose indegree and outdegree constraints on the nodes: i.e., we impose 
$x\bigl(\dt^\out(r)\bigr)=1=x\bigl(\dt^\into(t)\bigr)=x\bigl(\dt^\into(v)\bigr)=x\bigl(\dt^\out(v)\bigr)$
for all $v\neq r,t$, and $x\bigl(\dt^\into(r)\bigr)=0=x\bigl(\dt^\out(t)\bigr)$.

Consider the graph $G_k = (V_k,E_k)$ for a given value $k \geq 2$ shown in
Figure~\ref{fig:gap_2}, which is the standard example showing integrality gap of 
$\frac{3}{2}$ for the Held-Karp relaxation for symmetric \tsp. 
Here, $V_k=\{r, t, u_1, \ldots, u_{k}, v_1, \ldots, v_{k}\}$ 
and $E_k = \{ru_1, rv_1, u_1v_1\}\cup\{u_{k}t, v_{k}t, u_kv_k\}\cup
\bigcup_{i=1}^{k} \{u_iu_{i+1}, v_iv_{i+1}\}$.
All edges of $E_k$ have cost 1 and $c$ is the induced shortest-path metric.

\begin{figure}
\begin{center}
\includegraphics[width=0.4\textwidth]{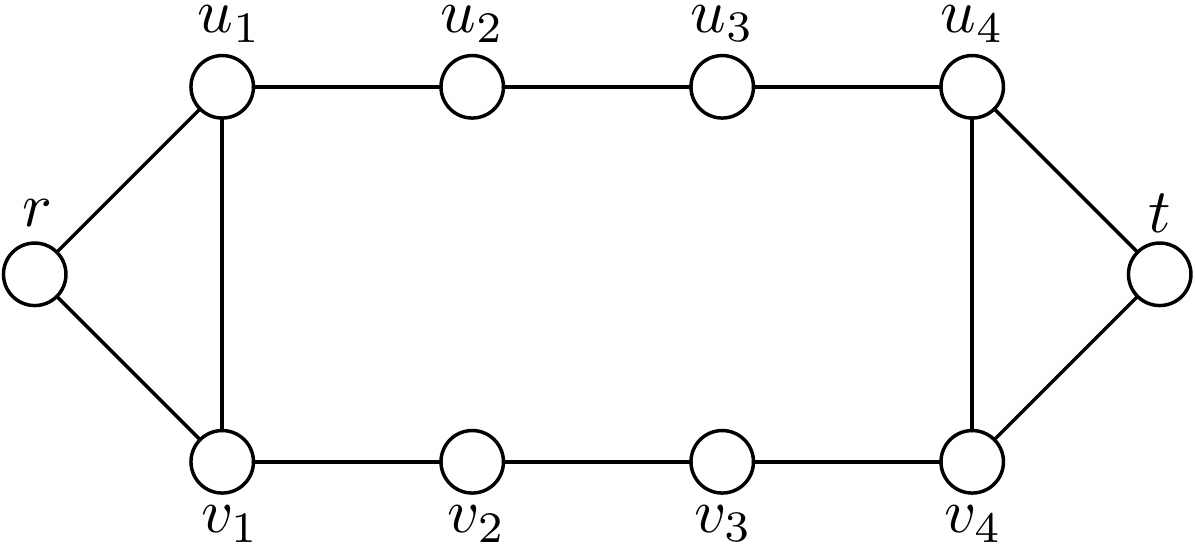}
\caption{The graph $G_4$.}\label{fig:gap_2}
\end{center}
\end{figure}

It is well known that any $r$-$t$ walk that visits all nodes has $c$-cost at least $3k$,
so since $\dist_t=k+1$, the optimal integer solution has $c^\reg$-cost at least $2k-1$. We
exhibit a fractional solution of $c^\reg$-cost $k$. Consider the following fractional solution.

\begin{align*}
& x_{r,u_1}=x_{r,v_1}=x_{u_k,t}=x_{v_k,t}=\tfrac{1}{2}, 
\qquad x_{u_1,v_1}=x_{v_1,u_1}=x_{u_k,v_k}=x_{v_k,u_k}=\tfrac{1}{4} \\
& \left.\begin{gathered}
x_{u_i,u_{i+1}}=x_{v_i,v_{i+1}}=\tfrac{3}{4} \\ 
x_{u_{i+1},u_i}=x_{v_{i+1},v_i}=\tfrac{1}{4} 
\end{gathered} \quad \right\} \quad \forall i=1,\ldots,k-1.
\end{align*}
All other $x_a$ are set to $0$.
It is easy to verify that this is a feasible fractional solution (satisfying the indegree
and outdegree constraints as well). Among the arcs $a$ with $x_a>0$, arcs $(u_1,v_1)$,
$(v_1,u_1)$, $(u_k,v_k)$, and $(v_k,u_k)$ have $c^\reg$-cost 1, and arcs of the form
$(u_{i+1},u_i)$ and $(v_{i+1},v_i)$ have $c^\reg$-cost 2. So we have 
$\sum_{a\in A}c^\reg_ax_a=k$. 
\end{proof}

\end{document}